\documentclass[a4paper,10pt]{article}
\usepackage[a4paper,hmargin=1.75cm,vmargin=2.5cm,bindingoffset=0.5cm]{geometry}
\usepackage{amsmath,amssymb}
\usepackage{amsthm}
\usepackage{mathtools}
\usepackage{mathrsfs}
\usepackage{setspace}
\usepackage{array}
\usepackage{subcaption}
\usepackage{microtype}
\usepackage[hidelinks]{hyperref}
\usepackage{algorithm}
\usepackage{algpseudocode}
\usepackage{tikz}
\usetikzlibrary{shapes, positioning}
\usetikzlibrary{decorations.pathreplacing}
\usetikzlibrary{arrows.meta}
\newtheorem{theorem}{Theorem}

\newtheorem{corollary}{Corollary}
\newtheorem{prop}{Proposition}

\newtheorem{claim}{Claim}[theorem]
\newtheorem{defn}{Definition}
\newcounter{case}
\newcounter{subcase}[case]

\newcommand{\Case}{%
  \refstepcounter{case}% allows referencing if needed
  \par\noindent \textit{Case \thecase.}\quad
}

\newcommand{\Subcase}{%
  \refstepcounter{subcase}%
  \par\noindent \textit{Case \thecase.\thesubcase.}\quad
}
\algtext*{EndIf}

\begin{document}
\title{Secure Domination in Bisplit graphs - A Structural and algorithmic study
%Short cycles dictate dichotomy status of the Steiner tree problem on Bisplit graphs
%Short cycles dictate dichotomy status of STREE on bisplit graphs
%Complexity of NP-Hard problems on bisplit graphs with short cycles
\thanks{This work is partially supported by NBHM-02011/24/2023/6051 and ANRF (DST)-CRG/2023/007127.} }
\author{Swathi D and
N Sadagopan\\
Indian Institute of Information Technology, Design and Manufacturing, Kancheepuram\\
\date{}
\texttt{cs24d0013@iiitdm.ac.in, sadagopan@iiitdm.ac.in}}
\maketitle
	\begin{abstract}
 A dominating set $S$ of a graph $G(V,E)$ is called a \textit{secure dominating set} if each vertex $u \in V(G) \setminus S$ is adjacent to a vertex $v \in S$ such that $(S \setminus \{v\}) \cup \{u\}$ is a dominating set of $G$. The \textit{secure domination number} $\gamma_s(G)$ of $G$ is the minimum cardinality of a secure dominating set of $G$. The \textit{Minimum Secure Domination problem} is to find a secure dominating set of a graph $G$ of cardinality $\gamma_s(G)$. In this paper, the computational complexity of  the secure domination problem on several graph classes is investigated. The decision version of secure domination problem was shown to be NP-complete on  star(comb) convex split graphs and bisplit graphs. So we further focus on complexity analysis of secure domination problem under additional structural restrictions on bisplit graphs. In particular, by imposing chordality as a parameter, we analyse its impact on the computational status of the problem on bisplit graphs. We establish the P versus NP-C dichotomy status of secure domination problem under restrictions on cycle length within bisplit graphs. In addition, we establish that the problem is polynomial-time solvable in chain graphs. We also prove that the secure domination problem cannot be approximated for a bisplit graph within a factor of $(1-\epsilon)~ln~|V|$ for any $\epsilon > 0$, unless $NP \subseteq DTIME(|V|^{O(log~log~|V|)})$.
 \end{abstract}
	\section{Introduction:}
	Domination and its variants in graph theory have been extensively studied due to their wide range of applications in areas such as computer networks, social networks, and locational studies. A practical scenario arises in securing a spatial domain such as a military terrain, university campus, or residential complex where guards must be positioned in such a way that every location is either directly guarded or can be  reached by a guard from a neighbouring location, without compromising on overall security. This setup can be modelled using graphs, where vertices represent locations and edges represent accessible paths for guard movement. The challenge is to identify a set of vertices such that every vertex is either in the set or adjacent to a vertex in the set, forming a dominating set. Additionally, for every vertex not in the set, there must be a neighbouring vertex in the set such that swapping their positions still results in a dominating set. The objective is to find such a minimum set of vertices. This precisely defines the $secure~domination~problem$.\\
Let $G = (V, E)$ be a graph. A subset $D \subseteq V(G)$ is a \textit{dominating set} of $G$ if every vertex in $V(G) \setminus D$ has at least one neighbour in $D$. The \textit{domination number} $\gamma(G)$ of $G$ is the minimum cardinality of a dominating set of $G$. The \textit{domination problem} is to find a minimum dominating set of a graph.
A variation of domination, \text{secure domination}, was introduced by Cockayne \textit{et al.}\cite{cockayne}. A dominating set $S$ of $G$ is called a \textit{secure dominating set} if each vertex $u \in V(G) \setminus S$ is adjacent to a vertex $v \in S$ such that $(S \setminus \{v\}) \cup \{u\}$ is a dominating set of $G$. The \textit{secure domination number} $\gamma_s(G)$ of $G$ is the minimum cardinality of a secure dominating set of $G$. The \textit{Minimum Secure Domination problem} $MSD$ is to find a secure dominating set of a graph $G$ of cardinality $\gamma_s(G)$.
Let  \textit{DD} denote the decision version of the domination problem and let \textit{SDD} denote the decision version of the secure domination problem.\\
\textit{Domination Problem (DD):}\\
\textit{Instance:} A graph $G = (V, E)$ and a positive integer $k \leq |V|$.\\
\textit{Question:} Does $G$ have a dominating set of cardinality at most $k$?\\
\textit{Secure Domination Problem (SDD):}\\
\textit{Instance:} A graph $G = (V, E)$ and a positive integer $k \leq |V|$.\\
\textit{Question:} Does $G$ have a secure dominating set of cardinality at most $k$?\\
The problem of secure domination was introduced by $Cockayne~et~al$., who investigated some fundamental properties of a secure dominating set, and also obtained exact values of  $\gamma_s(G)$ for some graph classes, such as paths, cycles and complete multipartite graphs\cite{cockayne}. Various properties and characterisations of secure domination set  have been researched \cite{Burger31032016, vdsd, li2017secure, merouane2015secure}. Upper and lower bounds on $\gamma_s(G)$ have been established for some graph classes \cite{tarakioutplan, tarakimaxoutplan, vdsd, li2017secure, merouane2015secure}. The problem of computing secure domination number is investigated for some restricted graph classes and their complexity status is found. The problem is NP-complete in general, and remains NP-complete when restricted to bipartite and split graphs \cite{merouane2015secure}, star convex bipartite graphs and doubly chordal graphs \cite{wang2018complexity}, and chordal bipartite graphs and undirected path graphs \cite{pradhan2018minimum}. On the positive side, the problem is linear on trees (subclass of bipartite) \cite{burger2014linear}. On subclasses of chordal graph such as block graphs \cite{pradhan2018minimum}, and proper interval graphs \cite{araki2018secure, tarakicorrecting} the problem is linear as well. The problem is linear in cographs\cite{araki2019secure, KISEK2021106155}. In addition, the problem is APX-complete for graphs with maximum degree 4 and there exists an inapproximability result for the problem \cite{wang2018complexity}.\\

\noindent $SDD$ is found to be NP-complete on split graphs. To impose convexity property on graph classes and to investigate the complexity status has been considered important in the literature. It would be interesting to study the complexity of $SDD$ in convex split graphs. In this paper we present the study of complexity status of the secure domination problem in star (comb) convex split graphs. In section \ref{sec:s_sd} we observe that the decision version of the secure domination problem is NP-complete for star(comb) convex split graphs. Since $SDD$ is found to be NP-complete on bipartite graphs and a $bisplit~graph$ is bipartite analog of a split graph. A natural direction of research is to investigate the computational complexity of the problem on bisplit graphs. In Section \ref{sec:bs_sd} we show that $SDD$ on bisplit graphs is NP-complete. An interesting dichotomy is obtained by imposing chordality as a structural parameter on bisplit graphs. In Section \ref{ssec:cbs}, we provide an algorithm for computing the secure domination of a $chordal ~bisplit~ graph$. In Section \ref{ssec:cbbs}, we show that $SDD$ is NP-complete for $chordal~bipartite~bisplit~graphs$. In Section \ref{chain:sd} we establish $MSD$ is polynomial-time solvable in chain graphs. We also prove that the secure domination problem cannot be approximated for a bisplit graph within a factor of $(1-\epsilon)~ln~|V|$ for any $\epsilon > 0$, unless $NP \subseteq DTIME(|V|^{O(log~log~|V|)})$ in Section \ref{ssec:inapp_SD_bs}, and then we establish the problem is linear time solvable for bounded tree-width graphs in Section \ref{btw_sd}.
\section{Secure Domination in Split graphs:}\label{sec:s_sd}	
\begin{defn}
A split graph $G(K,I)$ is called a $\pi-convex$ with $convexity~on~K$ if there is an associated structure $\pi$ on $K$ such that for each vertex $u\in I$, its neighbourhood $N_G(u)$ induces a connected subgraph in $\pi$.
\end{defn}
\begin{defn}
A split graph $G(K,I)$ is called a $\pi-convex$ with $convexity~on~I$ if there is an associated structure $\pi$ on $I$ such that for each vertex $u\in K$, its neighbourhood $N_G^I (u)$ induces a connected subgraph in $\pi$.
\end{defn}	
\noindent In general $\pi$ can be any arbitrary structure. Here we consider $\pi$ to be  $star$ and $comb$. Note that the structure $\pi$ in $G$ is an imaginary structure.\\~\\
	The construction used to reduce domination problem in split graphs to secure domination problem in split graphs  \cite{merouane2015secure} also generates star and comb convex instances. So we use the same construction to show that the \textbf{\textit{secure domination problem is NP-complete for star-convex split graphs with convexity on K, star-convex split graph with convexity on I and comb-convex split graphs with convexity on I.}} 
\\~\\ $ \textbf {\textit{Construction:}} $ Let $G=K \cup I$ be a split graph whose vertex set is partitioned into a clique $K$ and an independent set $I$. Let $G^*$ be constructed from $G$ by adding a path $P_2 : x - y$ such that $x$ is adjacent to $y$ and all the vertices of $G$. The graph thus obtained is a split graph $G^*=K^*  \cup I^*$ where $K^*=K \cup \{x\}$, $I^*=I \cup \{y\}$. We note that $|V(G^*)|=|V(G)|+2$ and $|E(G^*)|=|E(G)|+|V(G)|+1$. Thus $G^*$ can be constructed from $G$ in polynomial-time.
\begin{figure}[htbp]
	\begin{center}
\begin{tikzpicture}[scale=0.8, transform shape]
    % The set X
    \node[draw, ellipse, minimum width=1.5cm, minimum height=5cm] (K) at (0,0) {$K$};
    
    % The sets Y
    \node[draw, ellipse, minimum width=1.5cm, minimum height=5cm] (I) at (5,0) {$I$};
       
     % The set X*
    \node[draw, ellipse, minimum width=3cm, minimum height=8cm] (K*) at (0,0) {};
    \node at (0,-3.5) {$K^*$};
    \filldraw[black] (0,-2.8) circle (2pt) node[anchor=north](x) {$x$};
    % The sets Y*
    \node[draw, ellipse, minimum width=3cm, minimum height=8cm] (I*) at (5,0) {};
     \node at (5,-3.5) {$I^*$};
    \filldraw[black] (5,-2.8) circle (2pt) node[anchor=west](y) {$y$};
    % The set G*
    \node[below=of I, xshift=-2cm, yshift=-1cm] (G*) {$G^*$};
     % Connections
    \draw[blue,out=110, in=200] (-0.2,0.5) to (0,-2.8);
      \draw[blue, out=110, in=200] (-0.2,-0.5) to (0,-2.8);
  \draw[blue, out=120, in=200] (-0.2,1.5) to (0,-2.8); 
  \draw[blue, out=110, in=150] (-0.2,-1.5) to (0,-2.8);
    \draw[blue] (4.8,-0.5) edge (0,-2.8);
    \draw[blue](4.8,0.5) edge (0,-2.8);
  \draw[blue](4.8,1.5) edge (0,-2.8); 
  \draw[blue] (4.8,-1.5) edge (0,-2.8);
      \draw[blue] (5,-2.8) -- (0,-2.8);
\end{tikzpicture}
\end{center}
 \caption{$DD$ in split $\leq_p $ $SDD$ in split}
  \label{fig:splitredn}
\end{figure}
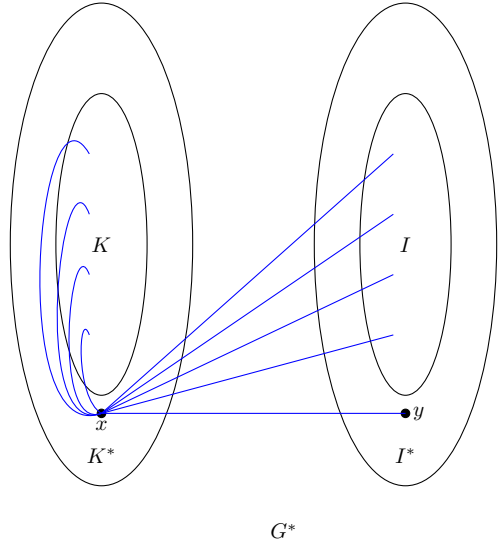
\begin{theorem}\label{th:1}
$SDD$ is NP-complete for star-convex split graphs with convexity on K.
\end{theorem}
\begin{proof}
$SDD~ is ~in ~NP:$ Given a graph $G(V,E)$ and a certificate $D \subseteq V (G)$, we show that there exists a deterministic polynomial-time algorithm for verifying the validity of $D$. Note that it is easy to check whether $|D| \leq k$. For each vertex in $V \setminus D$ we may scan through its adjacency list and ensure it has a vertex of $D$ this takes atmost $n^2$ comparisons, where $n$ is the cardinality of set $V$. Now for each vertex in $V \setminus D$ there are atmost $k$ neighbours in D. So to verify whether the swap set is dominating takes atmost $n^3$ steps. The certificate verification can be done in $O(n^3)$. Thus, we conclude that SDD is in NP.\\
$SDD~ is ~in ~NP-Hard:$ Let $G(K,I)$ be a split graph. It is known that $DD$ in split is NP-complete\cite{BERTOSSI198437}. $DD$ in split can be reduced in polynomial-time to $SDD$ on star-convex split graphs with convexity on $K^*$ using the construction in Figure \ref{fig:splitredn}. Note that $G$ is a split graph with $K$ being a clique and $I$ being an independent set. Now we show that $G^*$ is a star-convex split graph by defining an imaginary star $S$ on $K^*$.
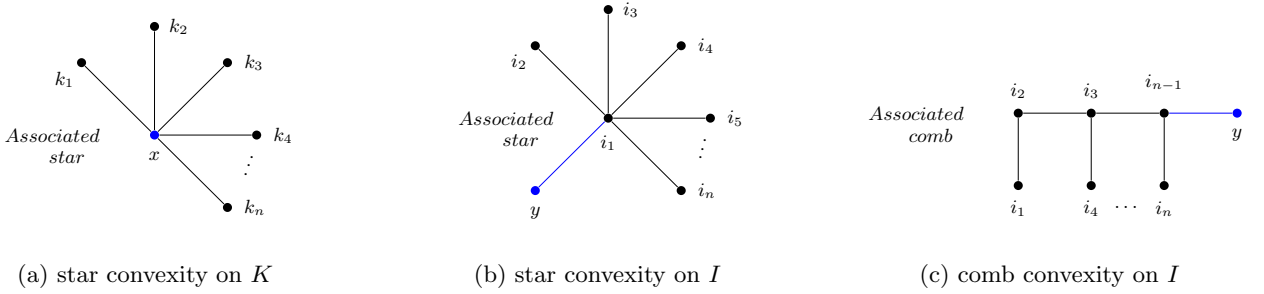
\begin{figure}[H]
    \centering
    \begin{subfigure}[b]{0.3\textwidth}
         \begin{center} \scalebox{0.8}{
\begin{tikzpicture}[scale=1.2, every node/.style={circle, fill=black, inner sep=1.5pt}, font=\small]

% Coordinates
\node[blue] (v5) at (0,0) {};
\node (v1) at (-1,1) {};
\node (v2) at (0,1.5) {};
\node (v3) at (1,1) {};
\node (v4) at (1,-1) {};
\node (v6) at (1.4,0) {};
% Edges
\draw (v5) -- (v1);
\draw (v5) -- (v2);
\draw (v5) -- (v3);
\draw (v5) -- (v4);
\draw (v5) -- (v6);

% Labels
\node[draw=none, fill=none, below left=1pt of v1]{$k_1$} ;
\node[draw=none, fill=none, right=1pt of v2] {$k_2$};
\node[draw=none, fill=none, right=2pt of v3] {$k_3$};
\node[draw=none, fill=none, right=2pt of v6] {$k_4$};
\node[draw=none, fill=none, right=2pt of v4] {$k_n$};
\node[draw=none, fill=none, rotate=70] at (1.3,-0.4) {$\dots$};
\node[draw=none, fill=none, below=2pt of v5] {$x$};

% Text label
\node[draw=none, fill=none]  at (-1.4,0){$Associated$};
\node[draw=none, fill=none] at (-1.2,-0.3) {$star$};

\end{tikzpicture}}
\end{center}

        \caption{star convexity on $K$}
         \label{fig:SconvsplitK}
    \end{subfigure}
    \hfill
    \begin{subfigure}[b]{0.3\textwidth}
         \begin{center} \scalebox{0.8}{
\begin{tikzpicture}[scale=1.2, every node/.style={circle, fill=black, inner sep=1.5pt}, font=\small]

% Coordinates
\node (v5) at (0,0) {};
\node (v1) at (-1,1) {};
\node (v2) at (0,1.5) {};
\node (v3) at (1,1) {};
\node (v4) at (1,-1) {};
\node (v6) at (1.4,0) {};
\node [blue] (v7) at (-1,-1) {};
% Edges
\draw (v5) -- (v1);
\draw (v5) -- (v2);
\draw (v5) -- (v3);
\draw (v5) -- (v4);
\draw (v5) -- (v6);
\draw [blue] (v5) -- (v7);
% Labels
\node[draw=none, fill=none, below left=1pt of v1]{$i_2$} ;
\node[draw=none, fill=none, right=1pt of v2] {$i_3$};
\node[draw=none, fill=none, right=2pt of v3] {$i_4$};
\node[draw=none, fill=none, right=2pt of v6] {$i_5$};
\node[draw=none, fill=none, right=2pt of v4] {$i_n$};
\node[draw=none, fill=none, below=2pt of v5] {$i_1$};
\node[draw=none, fill=none, below=2pt of v7] {$y$};
\node[draw=none, fill=none, rotate=75] at (1.3,-0.4) {$\dots$};

% Text label
\node[draw=none, fill=none]  at (-1.4,0){$Associated$};
\node[draw=none, fill=none] at (-1.2,-0.3) {$star$};

\end{tikzpicture}}
\end{center}
        \caption{star convexity on $I$}
         \label{fig:SconvsplitI}
    \end{subfigure}
    \hfill
    \begin{subfigure}[b]{0.3\textwidth}
      \begin{center} \scalebox{0.8}{
\begin{tikzpicture}[scale=1.2, every node/.style={circle, fill=black, inner sep=1.5pt}, font=\small]

% Coordinates
\node (v1) at (0,0) {};
\node (v2) at (0,-1) {};
\node (v3) at (1,0) {};
\node (v4) at (1,-1) {};
\node (v5) at (2,0) {};
\node (v6) at (2,-1) {};
\node [blue] (v7) at (3,0) {};
% Edges
\draw (v1) -- (v2);
\draw (v1) -- (v3);
\draw (v4) -- (v3);
\draw (v5) -- (v3);
\draw (v5) -- (v6);
\draw [blue] (v5) -- (v7);
% Labels
\node[draw=none, fill=none,  above=1pt of v1]{$i_2$} ;
\node[draw=none, fill=none, below=1pt of v2] {$i_1$};
\node[draw=none, fill=none, above=1pt of v3] {$i_3$};
\node[draw=none, fill=none, below=1pt of v4] {$i_4$};
\node[draw=none, fill=none, above=0.1 pt of v5] {$i_{n-1}$};
\node[draw=none, fill=none,  below=1pt of v6]{$i_n$} ;
\node[draw=none, fill=none, below=2pt of v7] {$y$};
\node[draw=none, fill=none] at (1.5,-1.3) {$\dots$};

% Text label
\node[draw=none, fill=none]  at (-1.4,0){$Associated$};
\node[draw=none, fill=none] at (-1.2,-0.3) {$comb$};

\end{tikzpicture}}
\end{center}
        \caption{comb convexity on $I$}
         \label{fig:CconvsplitI}
    \end{subfigure}
    \caption{Convexity on Split graphs}
    \label{fig:convsplit}
\end{figure}
\begin{claim}
$G^*$ is a star-convex split graph with convexity on $K^*$.
\end{claim}
\begin{proof}
We show that $G^*$ is a star-convex split graph by defining an imaginary star $S$ on $K^*$. Note that
for each $u \in I^*$ its neighbourhood $N_G(u)$ contains $x$. There exist an associated star $S$ on vertices of $K^*$ with  $x$ as root and other vertices of $K$ say $k_1, k_2, \dots k_n$ as leaves as shown in Figure  \ref{fig:SconvsplitK}. Therefore, for each $u \in I$, $N_G(u)$ is a subtree in S. Hence $G^*$ is a star-convex split graph with convexity on $K^*$.
\end{proof}
\begin{claim}\label{cl:split} \cite{merouane2015secure}
DD  $(G, k)$ if and only if SDD $(G^*, k+1)$.
\end{claim}
Thus we conclude $SDD$ is NP-Hard on the star-convex split graphs with convexity on $K$. Therefore, $SDD$ is NP-complete on star-convex split graphs with convexity on $K$.
\end{proof}
\begin{theorem}\label{th:2}
$SDD$ is NP-complete for star-convex split graphs with convexity on I.
\end{theorem}
\begin{proof}
$SDD~ is ~in ~NP-Hard:$ Let $G(K,I)$ be a split graph. It is known that $DD$ in star convex split with convexity on $I$ is NP-complete \cite{NSGstreeds}. $DD$ in star convex split graphs with convexity on $I$ can be reduced in polynomial-time to $SDD$ on star-convex split graphs with convexity on $I^*$ using the construction in Figure\ref{fig:splitredn}. Let $G$ be a star convex split graph with convexity on $I$. Let $T$ be the imaginary star associated with vertices of $I$. Now we show that $G^*$ is a star-convex split graph with convexity on $I^*$ by defining an imaginary star $T'$ on $I^*$.
\begin{claim}
$G^*$ is a star-convex split graph with convexity on $I^*$.
\end{claim}
\begin{proof}
We show that $G^*$ is a star-convex split graph by defining an imaginary star $T'$ on $I^*$. For each $v \in K^*$ its neighbourhood $N_G^{I^*}(v)$ remains the same except for $x \in K^*$. Note that $x$ is adjacent to all vertices of $I^*$. We may retain the same star $T$  associated with $I$ and add a  leaf $\{y \}$ to $T$. So, there exist an associated star $T' = T \cup \{ y \}$ on vertices of $I^*$ say $i_1, i_2, \dots i_n, y$ as shown in Figure \ref{fig:SconvsplitI}.Thus $G^*$ is also star-convex split graph with convexity on $I^*$. 
 \end{proof}
\noindent Also by Claim \ref{cl:split} we conclude $SDD$ is NP-Hard on the star-convex split graphs with convexity on $I$. Therefore, $SDD$ is NP-complete on star-convex split graphs with convexity on $I$.
\end{proof}
\begin{theorem}\label{th:3}
$SDD$ is NP-complete for comb-convex split graphs with convexity on I.
\end{theorem}
\begin{proof}
$SDD~ is ~in ~NP-Hard:$ Let $G(K,I)$ be a split graph. It is known that $DD$ in comb-convex split graphs with convexity on $I$ is NP-complete \cite{NSGstreeds}. $DD$ in comb convex split graphs with convexity on $I$ can be reduced in polynomial-time to $SDD$ on comb-convex split graphs with convexity on $I^*$ using the construction in Figure\ref{fig:splitredn}. Let $G$ be a comb convex split graph with convexity on $I$. Let $C$ be the imaginary comb associated with vertices of $I$. Now we show that $G^*$ is a comb-convex split graph with convexity on $I^*$ by defining an imaginary comb $C'$ on $I^*$.
\begin{claim}
$G^*$ is a comb-convex split graph with convexity on $I^*$.
\end{claim}
\begin{proof}
We show that $G^*$ is a comb-convex split graph by defining an imaginary comb $C'$ on $I^*$. For each $w \in K^*$ its neighbourhood $N_G^{I^*}(w)$ remains the same except for $x \in K^*$. Note that $x$ is adjacent to all vertices of $I^*$. We may retain the same comb $C$  associated with $I$ and add a vertex $\{y \}$ to $C$. Note that the vertex $y$ may be added to $C$ either as a $path~vertex$ or as a $teeth$, we see that comb-convexity is preserved. So, there exist an associated comb $C' = C \cup \{ y \}$ on vertices of $I^*$ say $i_1, i_2, \dots i_n, y$ as shown in Figure \ref{fig:CconvsplitI}.Thus $G^*$ is also comb-convex split graph with convexity on $I^*$. 
 \end{proof}
Also by Claim \ref{cl:split} we conclude $SDD$ is NP-Hard on the comb-convex split graphs with convexity on $I$. Therefore, $SDD$ is NP-complete on comb-convex split graphs with convexity on $I$.
\end{proof}
\begin{corollary}
$SDD$ is NP-complete for tree-convex split graphs
\end{corollary}
\begin{proof}
Since star-convex and comb-convex split graphs are a subclass of tree-convex split graphs, this result follows from Theorems \ref{th:1},\ref{th:2},\ref{th:3}.
\end{proof}
\section{Secure Domination in Bisplit graphs:}\label{sec:bs_sd}
\begin{defn}
	An undirected graph $G$ $=$ $(V,E)$ is a $bisplit$ $graph$ if its vertex set $V$ can be partitioned into three stable sets $X$, $Y$ and $Z$ such that $Y \cup Z$ induces a complete bipartite subgraph (a bi-clique) in $G$.
	\end{defn}
	\begin{prop}\label{p:epn_sd}\cite{cockayne} If $D$ is a secure dominating set of a graph $G$, then for every vertex $v \in D$, the subgraph induced by $epn(v, D)$ is complete.
	\end{prop}
\begin{theorem}\label{th:bs_sd}
$SDD$ is NP-complete for bisplit graphs.
\end{theorem}	
\begin{proof}
It is known that $DD$ in bisplit graphs is NP-complete \cite{MP_bisplit} and this can be reduced to $SDD$ in bisplit graphs using the following reduction.
	\\ \\ $ \textbf {\textit{Construction:}} $  Let $G=X \cup Y \cup Z$ be a bisplit graph whose vertex set is partitioned into three stable sets $X,~Y$ and $Z$ such that $Y \cup Z$ induces a complete bipartite graph. Let $G^*$ be obtained from $G$ as follows. We add a path $P_4 : x - y -z - x' $ such that, $x$ is adjacent to $y$; $y$ is adjacent to $x,~z$ and all vertices of $X$ and $Z$; $z$ is adjacent to $y,~x'$ and all vertices of $Y$ and $x'$ is adjacent to $z$. We note that $|V(G^*_)|=|V(G)|+4$ and $|E(G^*)|=|E(G)|+|X|+|Y|+|Z|+3 =|E(G)|+|V(G)|+3$. Thus $G^*$ can be constructed from $G$ in polynomial-time. For an illustration of this construction, we refer to Figure \ref{fig:bs_sd}.
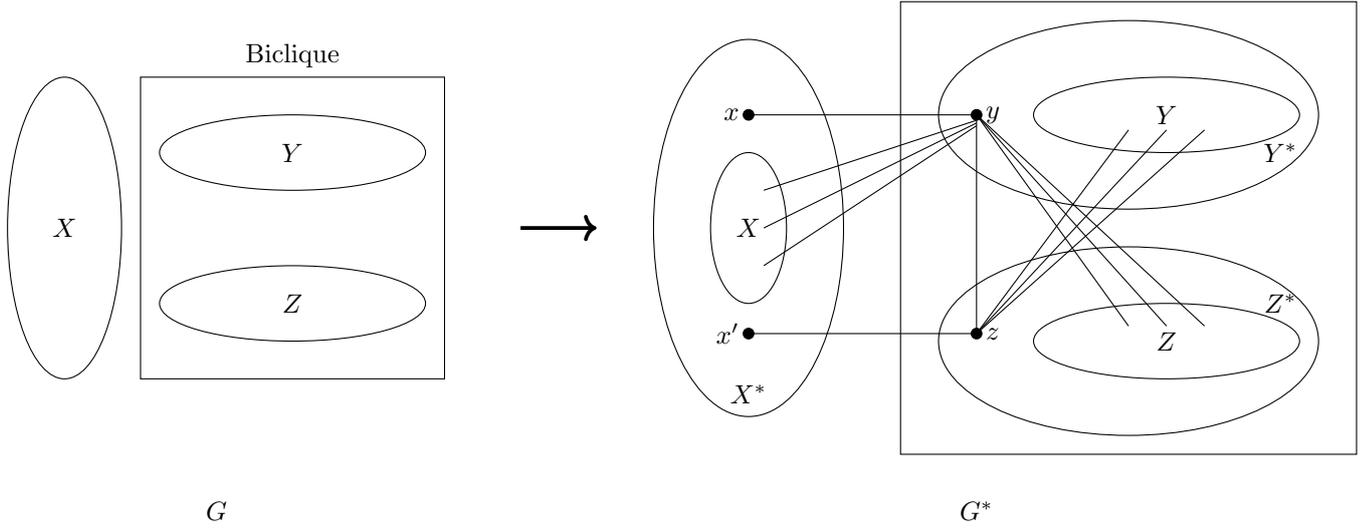
\begin{figure}[htbp]
\begin{center}
\begin{tikzpicture}
    % The set X
    \node[draw, ellipse, minimum width=1.5cm, minimum height=4cm] (X) at (0,0) {$X$};
    
    % The sets Y and Z
    \node[draw, ellipse, minimum width=3.5cm, minimum height=1cm] (Y) at (3,1) {$Y$};
    \node[draw, ellipse, minimum width=3.5cm, minimum height=1cm] (Z) at (3,-1) {$Z$};
    
    % The enclosing rectangle for the biclique
    \node[draw, rectangle, minimum width=4cm, minimum height=4cm, label=above:Biclique] (Biclique) at (3,0) {};
    
    % The set G
    \node[below=of Biclique, xshift=-1cm, yshift=-0.5cm] (G) {$G$};
 \draw [->, line width=0.5mm] (6,0) -- (7,0);
    % The set X* with x' marked
        \node[draw, ellipse, minimum width=2.5cm, minimum height=5cm] (Xstar) at (9,0) {};
    \node[draw, ellipse, minimum width=1cm, minimum height=2cm] (X) at (9,0) {$X$};
    \node at (9,-2.2) {$X^*$};
    \filldraw[black] (9,1.5) circle (2pt) node[anchor=east](x) {$x$};
    \filldraw[black] (9,-1.4) circle (2pt) node[anchor=east](x') {$x'$};
       
    % The sets Y*, Z* with labels
    \node[draw, ellipse, minimum width=5cm, minimum height=2.5cm] (Ystar) at (14,1.5) {};
    \node[draw, ellipse, minimum width=3.5cm, minimum height=1cm] (Y) at (14.5,1.5) {$Y$};
    \node at (16,1) {$Y^*$};
    \node[draw, ellipse, minimum width=5cm, minimum height=2.5cm] (Zstar) at (14,-1.5) {};
     \node[draw, ellipse, minimum width=3.5cm, minimum height=1cm] (Z) at (14.5,-1.5) {$Z$};
    \node at (16,-1) {$Z^*$};
     \filldraw[black] (12,1.5) circle (2pt) node[anchor= west](y) {$y$};
    \filldraw[black] (12,-1.4) circle (2pt) node[anchor= west](z) {$z$};
    
    % The enclosing rectangle for G*
    \node[draw, rectangle, minimum width=6cm, minimum height=6cm, ] (Box) at (14,0) {};
      \node[below=of Box, xshift=-2cm, yshift=0.5cm] (G*) {$G^*$};
    
    % Connections
    \draw (X)(9.2,0) edge (y)(12,1.5);
  \draw (X)(9.2,0.5) edge (y)(12,1.5); 
  \draw (X)(9.2,-0.5) edge (y)(12,1.5);
    \draw[] (x) -- (y);
\draw[] (y)(12,1.5) -- (12,-1.4);
\draw[] (y)(12,1.5) -- (14.5,-1.3);
\draw[] (y)(12,1.5) -- (14,-1.3);
\draw[] (y)(12,1.5) -- (15,-1.3);
\draw[] (z)(12,-1.4) -- (x');
\draw[] (z)(12,-1.4) -- (14.5,1.3);
\draw[] (z)(12,-1.4) -- (14,1.3);
\draw[] (z)(12,-1.4) -- (15,1.3);
\end{tikzpicture}
\end{center}
\caption{$DD$ in bisplit $\leq_p $ $SDD$ in bisplit}
 \label{fig:bs_sd}
\end{figure}
\begin{claim}\label{cl:bs_sd}
$G^*$ is a bisplit graph.
\end{claim}
\begin{proof}
Let $G^*=X^*  \cup Y^*  \cup Z^*$ where $X \subseteq X^*$, $Y \subseteq Y^*$ and $Z \subseteq Z^*$. We show that $G^*$ is a bisplit graph where $X^*, Y^*$ and $ Z^*$ are three stable sets and $Y^* \cup Z^*$ forms a biclique. By our construction, $y$ is adjacent to all vertices of $X$ and $Z$. So $y \in Y^*$. We know that $Y \cup Z$ of $G$ forms a biclique.   $x$ and $x'$ are neither adjacent to all vertices of $Y$ nor adjacent to all vertices of $Z$. These two vertices cannot be a part of biclique. $\{x, x'\} \in X^*$. Note that $z$ is adjacent to $y \in Y^*$, $x' \in X^*$ and all vertices of $Y$, which imply $z \in Z^*$.  Therefore, the graph $G^*=X^*  \cup Y^*  \cup Z^*$ where $X^*=X \cup \{x,x'\}$, $Y^*=Y \cup \{y\}$ and $Z^*= Z \cup \{z\}$. Note that $X^*, Y^*$ and $ Z^*$ are three stable sets and $Y^* \cup Z^*$ forms a biclique. Hence $G^*$ is a bisplit graph.
\end{proof}
\begin{claim}\label{cl:bs_npc_sd}
 $G$ has a dominating set $D$ with $ |D| \leq k $ if and only if $G^*$ has a secure dominating set $D^*$ with $|D^*| \leq k^* =k+2$.
\end{claim}
\begin{proof}
$\textbf {\textit{Yes instance of DD in G maps to Yes instance of SDD in G*:}} $
\\Let $G$ have 	a domination set $D$ of size atmost $k$. We shall prove that $G^*$ has a secure dominating set $D^*$ with $|D^*| \leq k^* =k+2$.\\
Consider $D^*=D \cup \{ y,z \}$. Clearly $D^*$ is a dominating set as $\{ y,z \}$ dominates all vertices of $G^*$ (by construction).
 Now we need to show $D^*$ is a secure dominating set of $G^*$. That is, we will prove for each vertex $v\in V(G^*)\setminus D^*$, there exists a neighbour $v^*\in D^*$ such that the swap set $(D^*\setminus\{ v^* \})\cup \{ v \} $ is again a dominating set of $G^*$
 \begin{itemize}
	\item When $v\in V(G)\setminus D$.\\
 Since $D$ is a subset of $D^*$ and $D$ is a dominating set of G, for every vertex $v\in V(G)\setminus D$, there exists a neighbour $v'\in D\subseteq D^*$. Therefore, for each $v\in V(G)\setminus D$, there exists a neighbour $v'\in D^*$ such that $(D^*\setminus\{ v' \})\cup \{ v \} $ is a dominating set of $G^*$. \\This is true as  $\{ y,z \}$ is a subset of the above set and it dominates all vertices of $G^*$.
 \item When $v=x$ there exists a neighbour $y\in D^*$ such that $(D^*\setminus\{ y \})\cup \{ x \} $ is a dominating set of $G^*$. \\This is true as  $D$, a subset of the above set dominates all vertices of $G$ and $\{ x,z \}$ in the above set dominates the newly introduced vertices.
 \item When $v=x'$ there exists a neighbour $z\in D^*$ such that $(D^*\setminus\{ z \})\cup \{ x' \} $ is a dominating set of $G^*$. \\This is true as  $D$, a subset of the above set dominates all vertices of $G$ and $\{ x',y \}$ in the above set dominates the newly introduced vertices.
 \end{itemize}
 Therefore, $D^*=D \cup \{ y,z \}$ is a secure dominating set of $G^*$ with $|D^*| \leq k^* =k+2$.\\
 
\noindent $\textbf {\textit{Yes instance of Secure domination in G* maps to Yes instance of  Domination in G:}} $
\\ Let $G^*$ have a secure dominating set $D^*$ with $|D^*|\leq k^*=k+2$. We consider the following cases depending on the value of $|D^* \cap \{ y,z\}|$ and we shall prove that $G$ has a dominating set $D$ with $|D|\leq k$ in each of these cases.\\

\noindent Consider $D'=D^* \cap V(G)$.  If $D'$ dominates $G$, then $D = D'$ and we are done. So let us assume that $D'$ is not a dominating set of $G$. Let $W$ be the non empty set of vertices of $G$ having no neighbour in $D'$. Let $W_x=W\cap X$, $W_y=W\cap Y$ and $W_z=W\cap Z$. Since $W$ is nonempty, at least one of the sets $W_x$, $W_y$ or $W_z$ is nonempty.\\
 \Case $|D^* \cap \{ y,z\}|=2$\\
\Subcase $|D^* \cap \{ x,x'\}|=2$. Since $D'$ doesnot dominate the vertices of $W \subseteq V(G)\setminus D^*$, each vertex of $W_x \cup W_z$ are dominated by $y \in G^*$ and each vertex of $W_y$ is dominated by $z \in G^*$. By Proposition \ref{p:epn_sd}, if $W_x \cup W_z \neq \phi$ then, $w \in W_x \cup W_z$ is sufficient to dominate it and if $W_y \neq \phi$ then $|W_y |=1$. Therefore $D=D' \cup \{w\} \cup W_y$ will be the dominating set of $G$ and we see that $|D|=|D^*|-|D^* \cap \{ y,z\}|-|D^* \cap \{ x,x'\}|+2 \leq k^*-2-2+2=k$.\\
\Subcase $|D^* \cap \{ x,x'\}|=1$. Here $|D'|=|D^*|-|D^* \cap \{ y,z\}|-|D^* \cap \{ x,x'\}| \leq k^*-2-1= k-1$. By our assumption since $D'$ doesnot dominate the vertices of $W \subseteq V(G)\setminus D^*$, we construct $D$ as follows:  
 \begin{itemize}
	\item If $x \in D^*$ by our case $x' \notin D^*$. Since $D^*$ is a secure domination set of $G^*$, $(D^*\setminus\{ z \})\cup \{ x' \} $ is a dominating set of $G^*$. Which imply all vertices of $Y$ are dominated by $D' \subset D^*$ and $W=W_x \cup W_z \neq \phi$. By Proposition \ref{p:epn_sd}, $epn(y, D^*)$ is complete, so $w \in W_x \cup W_z$ is sufficient to dominate it. Therefore $D=D' \cup \{w\}$ will be the dominating set of $G$ and $|D|=|D'|+1 \leq k$.
	 \item If $x' \in D^*$ by our case $x \notin D^*$. Since $D^*$ is a secure domination set of $G^*$, $(D^*\setminus\{ y \})\cup \{ x \} $ is a dominating set of $G^*$. Which imply all vertices of $X$ and $Z$ are dominated by $D' \subset D^*$ and $W=W_y \neq \phi$. By Proposition \ref{p:epn_sd}, $epn(z, D^*)$ is complete, so $|W_y |=1$. Therefore $D=D' \cup W_y$ will be the dominating set of $G$ and $|D|=|D'|+1 \leq k$.
\end{itemize}
\Subcase $|D^* \cap \{ x,x'\}|=0$. Since $x \notin D^*$ and $D^*$ is a secure domination set of $G^*$, $(D^*\setminus\{ y \})\cup \{ x \} $ is a dominating set of $G^*$. Which imply all vertices of $X$ and $Z$ are dominated by $D' \subset D^*$. Additionally, $x' \notin D^*$ and since $D^*$ is a secure domination set of $G^*$, $(D^*\setminus\{ z \})\cup \{ x' \} $ is a dominating set of $G^*$. Which imply all vertices of $Y$ are dominated by $D' \subset D^*$. Therefore, our assumption $D'$ is not a dominating set is wrong. Hence $D=D'$ is a dominating set of $G$ and $|D|=|D'| =|D^*|-|D^* \cap \{ y,z\}|-|D^* \cap \{ x,x'\}| \leq k^*-2+0=k$.\\
\Case $|D^* \cap \{ y,z\}|=1$\\
\Subcase $y \in D^*$. By our case $z \notin D^*$ which imply $x' \in D^*$ securely dominates $\{ x', z \}$ and all vertices of $Y$ are dominated by $D' \subset D^*$. So, $W = W_x \cup W_z$ and we construct $D$ as follows:  
 \begin{itemize}
	\item If $x \notin D^*$, since $D^*$ is a secure domination set of $G^*$, $(D^*\setminus\{ y \})\cup \{ x \} $ is a dominating set of $G^*$. Which imply all vertices of $X$ and $Z$ are dominated by $D' \subset D^*$. Therefore, our assumption $D'$ is not a dominating set is wrong. Hence $D=D'$ is a dominating set of $G$ and $|D|=|D'| =|D^*|-|D^* \cap \{ y,z\}|-|D^* \cap \{ x,x'\}| \leq k^*-1-1=k$.
	 \item If $ x \in D^*$ by our case $y \in D^*$, we conclude $epn(y, D^*)=W$. By Proposition \ref{p:epn_sd}, $epn(y, D^*)$ is complete, so $w \in W$ is sufficient to dominate it. Therefore $D=D' \cup \{w\}$ will be the dominating set of $G$ and $|D|=|D'|+1=|D^*|-|D^* \cap \{ y,z\}|-|D^* \cap \{ x,x'\}|+1 \leq k^*-1-2+1=k$.
\end{itemize}
\Subcase $z \in D^*$. By our case $y \notin D^*$ which imply $x \in D^*$ securely dominates $\{ x, y \}$ and all vertices of $X$ and $Z$ are dominated by $D' \subset D^*$. So, $W = W_y$ and we construct $D$ as follows:  
 \begin{itemize}
	\item If $x' \notin D^*$, since $D^*$ is a secure domination set of $G^*$, $(D^*\setminus\{ z \})\cup \{ x' \}$ is a dominating set of $G^*$. Which imply all vertices of $Y$ are dominated by $D' \subset D^*$. Therefore, our assumption $D'$ is not a dominating set is wrong. Hence $D=D'$ is a dominating set of $G$ and $|D|=|D'| =|D^*|-|D^* \cap \{ y,z\}|-|D^* \cap \{ x,x'\}| \leq k^*-1-1=k$.
	 \item If $ x' \in D^*$ by our case $z \in D^*$, we conclude $epn(z, D^*)=W$. By Proposition \ref{p:epn_sd}, $epn(z, D^*)$ is complete, so $|W_y|=1$. Therefore $D=D' \cup W_y$ will be the dominating set of $G$ and $|D|=|D'|+1=|D^*|-|D^* \cap \{ y,z\}|-|D^* \cap \{ x,x'\}|+1 \leq k^*-1-2+1=k$.
\end{itemize}
\Case $|D^* \cap \{ y,z\}|=0$\\
 $\Rightarrow$ $|D^* \cap \{ x,x'\}|=2$. That is $x$ and $x'$ securely dominates $\{y,x\}$ and $\{z,x'\}$ respectively. Which imply all vertices of $G$ are dominated by $D' \subset D^*$. Therefore, our assumption $D'$ is not a dominating set is wrong. Hence $D=D'$ is a dominating set of $G$ and $|D|=|D'| =|D^*|-|D^* \cap \{ y,z\}|-|D^* \cap \{ x,x'\}| \leq k^*-0-2=k$.\\
\\Hence in all the above cases $G$ has a dominating set $D$ with $|D|\leq k$.
\end{proof}
\noindent Thus from Claims \ref{cl:bs_sd} and \ref{cl:bs_npc_sd}, Theorem \ref{th:bs_sd} follows.
\end{proof}
\noindent A deeper inspection of Theorem \ref{th:bs_sd} reveals that the computational hardness of $SDD$ on bisplit graphs stems from the presence of cycles in the input graphs. This observation naturally motivates the following questions. First, what is the computational complexity of $SDD$ on bisplit graphs when the length of induced cycles is bounded—for instance, on chordal bisplit, chordal bipartite bisplit, and other special bisplit graphs? Second, can the maximum cycle length serve as a decisive structural parameter that dictates the polynomial-time versus NP-complete dichotomy of $SDD$ on bisplit graphs?
\subsection{Chordal Bisplit graphs}\label{ssec:cbs}
\begin{theorem}\cite{MP_bisplit}A graph $G=X \cup Y \cup Z$ is a $chordal~bisplit~graph$, if and only if the following properties are satisfied.\\
1. The biclique $Y \cup Z$ in $G$ is $K_{1,l}$ for some $l > 0$.\\
2. For each vertex $x \in X$, if $d_G (x) \geq 2$, then $x$ is adjacent to $y_1$.\\
3. The graph induced on $Z \cup X$ is a forest.
\end{theorem}
\noindent \textit{\textbf{Note:}\cite{cockayne}  For complete bipartite graph $K_{p,q}$ where $p \leq q$:} \[
\gamma_s =
\begin{cases}
q, & \text p =1,\\
2, & \text p = 2,\\
3, & \text p =3,\\
4, & \text p \geq 4.
\end{cases}
\]
	\textbf{The minimum secure domination problem is polynomial-time solvable on chordal bisplit graphs:}\\
Let the vertices of the bisplit graph $G=X \cup Y \cup Z$ be a labelled as follows:\\
$Y=\{y_1\}$; $Z=\{z_1, z_2, \dots, z_l\}$ where $Y \cup Z$ induces a biclique and $X=\{x_1, x_2, \dots, x_t\}$ forms and independent set. \\ 
\begin{algorithm}[H]
\caption{MSD: Chordal bisplit graphs}
\label{alg:cb-sdom}
\begin{algorithmic}
\Require A connected chordal bisplit graph $G=X \cup Y \cup Z$.
\Ensure A minimum secure dominating set $S$ of $G$.
\State Let $S'=X \cup Y$.
\State Let $Z'$ be the pendent vertices in Z. Let $Z' =\{z_1, z_2, \dots, z_g\}$.
\If{$Z' \neq \phi $.}
    \State  $S=S' \cup (Z' \setminus \{z_i\})$ for some $i; 1 \leq i \leq g$.
    \State \Return $S$.
     \EndIf
   \If { $Z' =  \phi $ and in $G''$, Case 1 or Case 4, or Case 5 are present from following observation as a substructure associated with $y_1$}.
    \State  $S=S' \setminus \{x_p\}$ for some $p; 1 \leq p \leq t$.
    \State \Return $S$
    \Else
     \State{$S=S'$}
      \EndIf
\end{algorithmic}
\end{algorithm}
\noindent \textbf {Observations:}\\
  Consider $G' = G \setminus \{Z' \cup \bigcup\limits_{k=1}^{m} x_k\} ; (1 \leq k \leq m \leq t)$ where $Z' =\{z_1, z_2, \dots, z_g\}$ is the set of  pendant vertices in $Z$ and $x_k$ are pendant vertices in $X$. Let $G'' = G' \setminus Z''$ where  $Z''$ is set of pendant vertices in $G'$. We observe $G''$  have the following structures associated with the vertex $y_1$.\\
  \begin{figure}[H]
  \begin{center}
\begin{tikzpicture}[scale=.8, line width=1pt]
\tikzset{
    v/.style={circle,draw,inner sep=2pt}
}

%%%%%%%%%%%%%%%%%%%%%%%%%%%%%%
%%% CASE 1 %%%
%%%%%%%%%%%%%%%%%%%%%%%%%%%%%%
\begin{scope}[shift={(0,0)}]
\node[v,label=left:{$z_i$}]  (A) at (0,0) {};
\node[v,fill=red,label=right:{$x_p$}] (B) at (1,-1.5) {};
\node[v,fill=blue,label=above:{$y_1$}] (C) at (1,1.5) {};
\draw (A)--(B)--(C)--(A);

\node at (1,-2.2) {\textbf{Case 1}};
\end{scope}

%%%%%%%%%%%%%%%%%%%%%%%%%%%%%%
%%% CASE 2 %%%
%%%%%%%%%%%%%%%%%%%%%%%%%%%%%%
\begin{scope}[shift={(6,0)}]
\node[v,label=left:{$z_i$}]  (D) at (0,0) {};
\node[v,label=right:{$z_j$}] (E) at (2,0) {};
\node[v,fill=blue,label=above:{$y_1$}] (F) at (1,1.5) {};
\node[v,fill=blue,label=below:{$x_k$}] (G) at (1,-1.5) {};

\draw (F)--(E)--(G)--(D)--(F)--(G);

\node at (1,-2.2) {\textbf{Case 2}};
\end{scope}

%%%%%%%%%%%%%%%%%%%%%%%%%%%%%%
%%% NEW CASE 3  %%%
%%%%%%%%%%%%%%%%%%%%%%%%%%%%%%
\begin{scope}[shift={(12,0)}]
\node[v,label=left:{$z_i$}]  (H) at (-1.5,0) {};
\node[v,label=right:{$z_j$}] (I) at (0,0) {};
\node[v,label=left:{$z_k$}]  (J) at (2,0) {};
\node[v,label=right:{$z_l$}] (K) at (3.5,0) {};
\node[v,fill=blue,label=above:{$y_1$}] (L) at (1,1.5) {};
\node[v,fill=blue,label=below:{$x_m$}] (M) at (1,-1.5) {};

\draw (M)--(H)--(L)--(I)--(M)--(J)--(L)--(K)--(M)--(L);
\node at (1,-2.2) {\textbf{Case 3}};
\end{scope}

%%%%%%%%%%%%%%%%%%%%%%%%%%%%%%
%%% NEW CASE 4  %%%
%%%%%%%%%%%%%%%%%%%%%%%%%%%%%%
\begin{scope}[shift={(3,-4.5)}]
\node[v,label=left:{$z_i$}]  (H) at (-2.5,0) {};
\node[v,label=right:{$z_j$}] (I) at (-0.75,0) {};
\node[v,label=left:{$z_k$}]  (J) at (0.75,0) {};
\node[v,label=right:{$z_l$}] (K) at (2.5,0) {};
\node[v,fill=blue,label=above:{$y_1$}] (L) at (0,1.5) {};
\node[v,fill=red,label=below:{$x_p$}] (N) at (-3.1,-1.5) {};
\node[v,fill=blue,label=below:{$x_q$}] (M) at (0,-1.5) {};
\node[v,fill=blue,label=below:{$x_r$}] (O) at (3.1,-1.5) {};

\draw (H)--(L)--(K)--(O)--(J)--(M)--(I)--(N)--(H);
\draw (L)--(N);
\draw (L)--(M);
\draw (L)--(O);
\draw (L)--(I);
\draw (L)--(J);
\node at (0,-2.2) {\textbf{Case 4}};
\end{scope}

%%%%%%%%%%%%%%%%%%%%%%%%%%%%%%
%%% NEW CASE 5  %%%
%%%%%%%%%%%%%%%%%%%%%%%%%%%%%%
\begin{scope}[shift={(11,-4.5)}]
\node[v,label=left:{$z_i$}]  (H) at (-2.5,0.3) {};
\node[v,label=right:{$z_j$}] (I) at (0,-1.3) {};
\node[v,label=left:{$z_k$}]  (J) at (2.5,0) {};
\node[v,label=right:{$z_l$}] (K) at (3.8,0.5) {};
\node[v,fill=blue,label=above:{$y_1$}] (L) at (0,1.5) {};
\node[v,fill=red,label=below:{$x_p$}] (N) at (-2,-2.3) {};
\node[v,fill=blue,label=below:{$x_q$}] (M) at (2,-2.5) {};
\node[v,fill=blue,label=below:{$x_r$}] (O) at (3.3,-2.1) {};

\draw (H)--(L)--(K)--(O)--(I)--(M)--(I)--(N)--(H);
\draw (L)--(N);
\draw (L)--(M);
\draw (L)--(O);
\draw (L)--(I);
\draw (L)--(J);
\draw (M)--(J);
\node at (0,-2.2) {\textbf{Case 5}};
\end{scope}

\end{tikzpicture}
\end{center}
   \caption{Structures associated with the vertex $y_1$}
         \label{fig:cbs_sd}	
\end{figure}
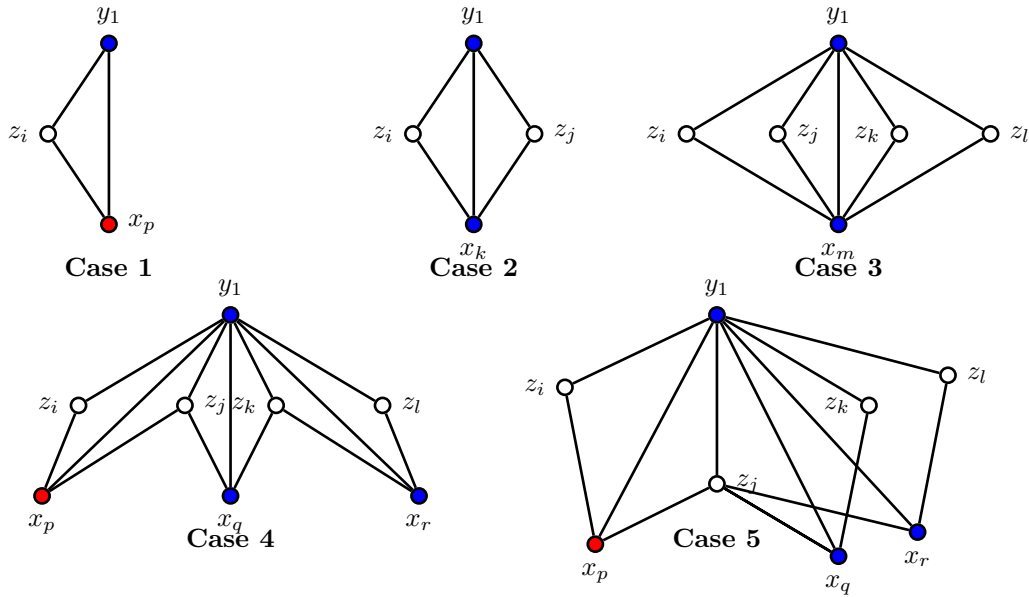
\subsection{Chordal Bipartite Bisplit Graphs}\label{ssec:cbbs}
We observed that $SDD$ was polynomial-time solvable on chordal bisplit graphs. This motivates us to analyse the complexity of $SDD$ on bisplit graphs having a cycle of length four, which are precisely the chordal bipartite bisplit graphs. In this section, we show that $SDD$ on chordal bipartite bisplit graphs  is NP-hard by presenting a polynomial-time reduction. 
\begin{theorem}\label{th:cbbs_sd}
$SDD$ is NP-complete for chordal bipartite bisplit graphs.
\end{theorem}	
\begin{proof}
The candidate problem for reduction is $SDD$ on $chordal~bipartite~graphs$ which is known to be NP-complete\cite{pradhan2018minimum}.\\
 $ \textbf {\textit{Construction:}} $	
 Let $G=(X \cup Y,E)$ be a chordal bipartite graph with bipartitions $X$ and $Y$. Let vertex set $X = \{k_1, k_2, \dots, k_p \}$ and $Y=\{l_1, l_2, \dots, l_q \}$. We construct $G'= (X_1 \cup X_2 \cup Y_1 \cup Y_2,E')$ from $G$ as follows. The vertex set $V(G')$ consists of original vertices of $G$, a copy of vertices of $G$ labelled as $\{m_1, m_2, \dots, m_p;$ $n_1, n_2, \dots, n_q \}$ and four additional vertices $k, l, m, n$. Define $X_1=\{k, k_i | k_i \in X; 1\leq i \leq p \}$; $Y_1=\{l, l_i | l_i \in Y; 1\leq i \leq q \}$; $X_2=\{m, m_i | k_i \in X; 1\leq i \leq p \}$ and $Y_2=\{n, n_i | l_i \in Y; 1\leq i \leq q \}$. The edge set $E'$ is defined as $E'=\{E(G)\} \cup \{\{m_i,n_j\} | \{k_i,l_j\} \in E; 1\leq i \leq p, 1 \leq j \leq q\} \cup \{\{k_i,m_j\} | \forall i,j; 1\leq i \leq p, 1 \leq j \leq q\}\cup \{E^*\}$ where $E^*=\{\{l,k\}, \{k,m\}, \{m,n\}, \{k,m_i\}, \{m,k_i\} | \forall i; 1\leq i \leq p\}$. We note that, $|V(G')|=2 \cdot |V(G)| +4$ and $|E'|=2 \cdot |E(G)|+(p+1)^2 +2$.
	\\Thus $G'$ can be constructed in polynomial-time. For an illustration of this construction, we refer to Figure \ref{fig:cbbs_sd}.

	\begin{figure}[htbp]
	\begin{center}
\begin{tikzpicture}[scale=0.7, transform shape]
   
    % The set X* with x' marked
        \node[draw, ellipse, minimum width=4cm, minimum height=10cm] (X') at (7,0) {};
    \node[draw, ellipse, minimum width=2cm, minimum height=4.2cm] (I1) at (7,2.2) {};
      \node[draw, ellipse, minimum width=2cm, minimum height=4.2cm] (I2) at (7,-2.2) {};
    \node at (7,-4) {$Y_2$};
    \node at (7,4) {$Y_1$};
    \filldraw[black] (7,3) circle (2pt) node[anchor=east](x) {$l$};
    \filldraw[black] (7,-3) circle (2pt) node[anchor=east](x') {$n$};
       
    % The sets Y*, Z* with labels
    \node[draw, ellipse, minimum width=8cm, minimum height=4cm] (Y') at (16,2.5) {};
    \node[draw, ellipse, minimum width=5cm, minimum height=2cm] (Y) at (16.5,2.5) {};
    \node at (19.5,2.5) {$X_1$};
    \node at (18.7,2.5) {};
    \node[draw, ellipse, minimum width=8cm, minimum height=4cm] (Z') at (16,-2.5) {};
     \node[draw, ellipse, minimum width=5cm, minimum height=2cm] (Z) at (16.5,-2.5) {};
    \node at (19.5,-2.5) {$X_2$};
     \node at (18.7,-2.5) {};
     \filldraw[black] (13,3) circle (2pt) node[anchor= west](y) {$k$};
    \filldraw[black] (13,-3) circle (2pt) node[anchor= west](z) {$m$};
    %nodes of G
    \filldraw[black] (7.3,2) circle (1pt) node[anchor= east](x3) {$l_2$};
      \filldraw[black] (7.3,0.8) circle (1pt) node[anchor= east](x2) {$l_q$};
             \filldraw[black] (7.3,1.2) circle (1pt);
               \filldraw[black] (7.3,1.4) circle (1pt);
	\filldraw[black] (7.3,1.6) circle (1pt);
       \filldraw[black] (7.3,2.4) circle (1pt) node[anchor= east](x1) {$l_1$};
       
       \filldraw[black] (7.3,-2) circle (1pt) node[anchor= east](x3') {$n_2$};
      \filldraw[black] (7.3,-0.8) circle (1pt) node[anchor= east](x2') {$n_q$};
            \filldraw[black] (7.3,-1.2) circle (1pt);
  \filldraw[black] (7.3,-1.4) circle (1pt);
	\filldraw[black] (7.3,-1.6) circle (1pt);
       \filldraw[black] (7.3,-2.4) circle (1pt) node[anchor= east](x1') {$n_1$};
              
     \filldraw[black] (15,2.5) circle (1pt) node[anchor= south](c1) {$k_1$};
      \filldraw[black] (16,2.5) circle (1pt) node[anchor= south](c2) {$k_2$};
       \filldraw[black] (16.8,2.5) circle (1pt);
       \filldraw[black] (17,2.5) circle (1pt);
	\filldraw[black] (17.2,2.5) circle (1pt);
       \filldraw[black] (18,2.5) circle (1pt) node[anchor= south](c3) {$k_p$};
    
    \filldraw[black] (15,-2.5) circle (1pt) node[anchor= north](c1') {$m_1$};
\filldraw[black] (16,-2.5) circle (1pt) node[anchor= north](c2) {$m_2$};
\filldraw[black] (16.8,-2.5) circle (1pt);
       \filldraw[black] (17,-2.5) circle (1pt);
	\filldraw[black] (17.2,-2.5) circle (1pt);
       \filldraw[black] (18,-2.5) circle (1pt) node[anchor= north](c3) {$m_p$};
    % The enclosing rectangle for G*
    \node[draw, rectangle, minimum width=9cm, minimum height=10cm,] (Box) at (16,0) {};
      \node[below=of Box, xshift=-3.5cm, yshift=0.5cm] (G') {$G'$};
    \draw[decorate, decoration={brace,mirror}](7.5,0.8)--(7.5,2.4);
     \draw[decorate, decoration={brace}](7.5,-0.8)--(7.5,-2.4);
    % Connections
    \draw (7,3)--(13,3)--(13,-3)--(7,-3);
\draw[transparent] (13,2.5)--(7.3,2.4);
      \draw [transparent](13,2.5)--(7.3,0.8);
       \draw[dotted, thick] (14,2.5)--(7.6,1.6);
       \draw (13,3)--(15,-2.5);
 	\draw (13,3)--(16,-2.5);
 	\draw (13,3)--(18,-2.5);
\draw[dotted, thick] (14,-2.5)--(7.6,-1.6);
      \draw[transparent] (13,-2.5)--(7.3,-0.8);
       \draw [transparent](13,-2.5)--(7.3,-2);
       \draw (13,-3)--(15,2.5);
 	\draw (13,-3)--(16,2.5);
 	\draw (13,-3)--(18,2.5);
 \end{tikzpicture}
\end{center}
\caption{$SDD$ in chordal bipartite $\leq_p $ $SDD$ in chordal bipartite bisplit}
 \label{fig:cbbs_sd}

\end{figure}
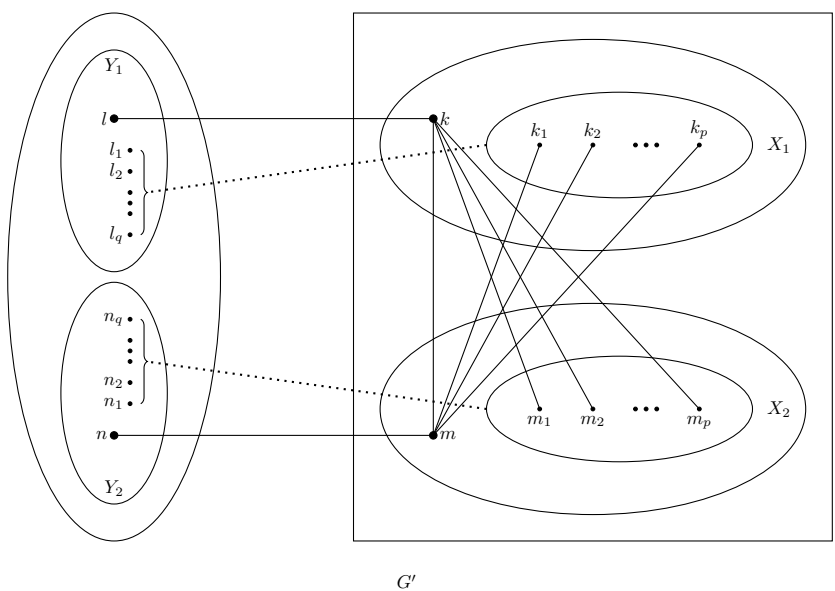
	\begin{claim}\label{bbs:cbbs_sd}
$G'$ is a bipartite bisplit graph.
\end{claim}
\begin{proof}
We first show that $G'=X' \cup Y' \cup Z'$ is a bisplit graph where $X, Y$ and $ Z$ are three stable sets and $Y \cup Z$ forms a biclique. By our construction, $X_1 \cup X_2$ induces a biclique. Also, no two vertices of $Y_1 \cup Y_2$ are adjacent; hence it induces and independent set. Therefore, the graph $G'=X' \cup Y' \cup Z'$ where $X'=Y_1 \cup Y_2$, $Y'=X_1$ and $Z'=X_2$ forms a bisplit graph. Further, by our construction vertices of $Y_1$ are non adjacent to vertices of $X_2$ and vertices of $Y_2$ are non adjacent to vertices of $X_1$. Thus $A=X_1 \cup Y_2$ and $B=X_2 \cup Y_1$ forms a bipartition for $G'$. This completes the proof.
 \end{proof}
\begin{claim}\label{cb:cbbs_sd}
$G'$ is chordal bipartite.
\end{claim}
\begin{proof}
Since $G$ is chordal bipartite graph, the subgraphs induced by $X_1 \cup Y_1$ and $X_2 \cup Y_2$ are chordal bipartite. So, to prove $G'$ is chordal bipartite is suffices to show that, every cycle $\mathscr{C}$ of length six involving vertices of both $X_1 \cup Y_1$ and $X_2 \cup Y_2$ has four vertices of biclique forming a $C_4$.\\ Assume that $\mathscr{C}$  have fewer than four vertices from biclique. Let $\mathscr{C}$ contain three vertices from the biclique $X_1 \cup X_2$. These three vertices induces a $P_3$ in $G'$. By our construction, if the pendants of $P_3$ are in $X_2$ then the other three vertices of $\mathscr{C}$ are from $Y_2$ and if the pendants of $P_3$ are in $X_1$ then the other three vertices of $\mathscr{C}$ are from $Y_1$. Also, by Claim \ref{bbs:cbbs_sd} $G'=A \cup B$ is bipartite. So, such a cycle $\mathscr{C}$ involving four vertices from one of the partitions and 2 vertices from other is not possible. By similar reasoning we may show that $\mathscr{C}$ cannot contain fewer than three vertices from the biclique. Thus the assumption fails and consequently the claim follows.
\end{proof}
\begin{claim}\label{npc:scs;cbbs_sd}
 $G$ has a secure dominating set $S$ with $ |S| \leq k $ if and only if $G'$ has a secure dominating set $S'$ with $|S'| \leq 2k+2$.
\end{claim}
\begin{proof}
$\textbf {\textit{Yes instance of SDD in G maps to Yes instance of SDD in G':}} $
\\Let $G$ have 	a secure dominating set $S$ of size atmost $k$. We shall prove that $G'$ has a secure dominating set $S'$ with $|S'| \leq 2k+2$.\\
Define $G^*=G' \setminus \{k, l, m, n\}$. Consider the set $S^*=S_1 \cup S_2$ where $S_1=\{k_i, l_j | k_i \in S, l_j \in S; 1\leq i \leq p, 1 \leq j \leq q\}$ and $S_2 = \{ m_i, n_j | k_i \in S, l_j \in S;  1\leq i \leq p, 1 \leq j \leq q\}$. Let $S'=S ^* \cup \{ k, m \}$ and $|S'| \leq 2k+2$. Clearly $S'$ is a dominating set of $G'$ as $S^*$ dominates all vertices of $G^*$ and $\{ k, m \} \in S'$ dominates $\{k, l, m, n\}$.
 Now we need to show $S'$ is a secure dominating set of $G'$. That is, we will prove for each vertex $v\in V(G')\setminus S'$, there exists a neighbour $v'\in S'$ such that the swap set $(S' \setminus\{ v' \})\cup \{ v \} $ is again a dominating set of $G'$.
 \begin{itemize}
 \item When $v=l$ there exists a neighbour $k\in S'$ such that $(S'\setminus\{ k \})\cup \{ l \} $ is a dominating set of $G'$. \\This is true as  $S^*$, a subset of the above set dominates all vertices of $G^*$ and $\{ l, m \}$ in the above set dominates the four other vertices. 
 \item When $v=n$ there exists a neighbour $m \in S'$ such that $(S' \setminus \{ m \})\cup \{ n \} $ is a dominating set of $G'$. \\This is true as  $S^*$, a subset of the above set dominates all vertices of $G^*$ and $\{ k, n \}$ in the above set dominates the four other vertices. 
 \item When $v\in V(G^*)\setminus S^*$.\\Since $S$ is a secure dominating set of $G$, $S_1$ securely dominates vertices of $Y_1$ as adjacencies in $G$ are preserved in $G'$. As a consequence, $S^*$ securely dominates vertices of $Y_1 \cup Y_2$. It suffices to show that, $v\in (V(G^*)\setminus S^*) \cap (X_1 \cup X_2)$ is securely dominated. Since $S^*$ is a subset of $S'$ and $S^*$ is a dominating set of $X_1 \cup X_2$, for every vertex $v\in (V(G^*)\setminus S^*) \cap (X_1 \cup X_2)$, there exists a neighbour $v' \in S^* \subseteq S'$. Therefore, for each $v \in (V(G^*)\setminus S^*) \cap (X_1 \cup X_2)$, there exists a neighbour $v'\in S'$ such that $(S' \setminus\{ v' \})\cup \{ v \} $ is a dominating set of $G'$. \\This is true as  $\{ k, m \}$ is a subset of the above set and it dominates all vertices of $X_1 \cup X_2$.
  \end{itemize}
 Therefore, $S'=S^* \cup \{ k, m \}$ is a secure dominating set of $G'$ with $|S'| \leq 2k+2$.\\
 $\textbf {\textit{Yes instance of SDD in G' maps to Yes instance of  SDD in G:}} $
\\ Let $G'$ have a secure dominating set $S'$ with $|S'|\leq 2k+2$. We shall prove that $G$ has a secure dominating set $S$ with $|S|\leq k$. We note that, since $l$ and $n$ are pendant vertices of $G'$, $|S' \cap \{k, l\}| \geq 1$ and $|S' \cap \{m, n\}| \geq 1$. Let $G^*=G' \setminus \{k, l, m, n\}$. Consider $S^*= S' \cap V(G^*)$. $|S^*| = |S'|-[|S' \cap \{k, l\}| + |S' \cap \{m, n\}|] \leq (2k+2)-2 = 2k$. If $S^*$ securely dominates $G^*$, then we are done. Because, $S = S^* \cap (X_1 \cup Y_1)$ forms a secure dominating set of $G$ and since $G^*$ is symmetric $|S| = |S^*| / 2 \leq k$. So, we assume $S^*$ is not a secure dominating set of $G^*$. Let $W \subseteq V(G^*)$ be the set of vertices that are not securely dominated by $S^*$. But, $W$ was securely dominated by $S'$ in $G'$. This implies that $W$ was securely dominated be vertices of $S' \cap \{k, l, m, n\}$ in$G^*$. Thus, by construction, $W \subseteq X_1 \cup X_2$. Let  $W_1=W \cap X_1$ and $W_2=W \cap X_2$. Without loss of generality, assume $W_1\neq \emptyset$. We know that $W_1$ is securely dominated by $S'$ in $G'$, which imply $m \in S'$. By Proposition \ref{p:epn_sd}, $epn(m, S')$ is complete. It follows that $W_1$ induces a complete subgraph of $G'$ and $n \in S'$. Consequently, $|S' \cap \{m,n\}|=2$. Further, since $X_1$ is an independent set, we have $|W_1|=1$. Now, $W_2 \neq \emptyset$ or $W_2= \emptyset$. If $W_2 \neq \emptyset$, by similar argument it can be shown that $|W_2|=1$ and $|S' \cap \{k,l\}|=2$. If $W_2 = \emptyset$, then $|W_2|=0$ and $|S' \cap \{k,l\}|\geq 1$. From the above, $|W|=|W_1|+|W_2| \leq |S' \cap \{k, l\}| + |S' \cap \{m, n\}|-2$. It follows that $|S^*| = |S'|-[|S' \cap \{k, l\}| + |S' \cap \{m, n\}|] \leq |S'|-[|W|+2]$. Also, $S^* \cup W$ is a secure dominating set of $G^*$ and $|S^* \cup W| = |S^*|+|W| \leq |S'|-[|W|+2]+|W|=|S'|-2 \leq 2k$. Further, since $G^*$ is symmetric $S = (S^*\cup W) \cap (K_1 \cup I_1)$ forms a secure dominating set of $G$ and $|S| = |S^* \cup W| / 2 \leq k$.
\end{proof}
\noindent Thus from Claims \ref{bbs:cbbs_sd}, \ref{cb:cbbs_sd}  and \ref{npc:scs;cbbs_sd}, Theorem \ref{th:cbbs_sd} follows.
\end{proof}
\noindent Therefore, secure domination problem is hard on chordal bipartite bisplit graphs. An interesting  line of investigation is to study  its complexity on $P_5~free~chordal~bipartite~graphs$ precisely the chain graphs.
\subsection{Chain graphs}\label{chain:sd}
In this section, we present an algorithm to compute the minimum secure dominating set of a chain graph.
\begin{defn}
 A bipartite graph $G = (X, Y, E)$ is a chain graph, if there exists a chain ordering of $X \cup Y$ , i.e. ($x_1, x_2, \dots, x_n, y_1, y_2, \dots, y_m$) such that $N(x_1) \subseteq N(x_2) \subseteq \dots \subseteq N(x_n)$ and $N(y_1) \supseteq N(y_2) \supseteq \dots \supseteq N(y_m)$. 
 \end{defn}
\noindent The chain ordering of a chain graph can be computed in linear-time \cite{pinar_et.al}. Define a relation $R$ on $X$ as follows. Let $x_i$ and $x_j$ are related if $N(x_i) = N(x_j)$. Observe that $R$ is an equivalence relation. Assume that $X_1, X_2, \dots, X_k$ is the partition of $X$ based on the relation $R$. Define $Y_1 = N(X_1)$ and $Y_i = N(X_i) \setminus \bigcup_{ j=1}^{i-1} N(X_j)$; $i = 2, 3, \dots k$. Then,$Y_1, Y_2, \dots, Y_k$ forms a partition of $Y$. Such partition $X_1, X_2, \dots, X_k, Y_1, Y_2, \dots,Y_k$ of $X \cup Y$ is called a $proper~ ordered~chain~partition$ of $X \cup Y$. Note that the number of sets in the partition of $X$ and $Y$ are same. Further, the set of pendent vertices of $G$ are contained in the set $X_1 \cup Y_k$. Let $X_1' \subseteq X_1$ and $Y_k' \subseteq Y_k$ denote the set of pendant vertices in $G$.\\
\begin{algorithm}[H]
\caption{MSD: Chain graphs}
\label{alg:chain_sd}
\begin{algorithmic}
\Require A connected chain graph $G=(X \cup Y ,E)$ with proper order chain partition $X_1, X_2, \dots, X_k$ and $Y_1, Y_2, \dots, Y_k$ of $X$ and $Y$.
\Ensure A minimum secure dominating set $S$ of $G$.
\State Let $S'= \{x_{m-1}, x_m, y_1, y_2\}$.
\If{$X_1' \geq 2$ where $x_1 \in X_1'$}
\If{$Y_k' \geq 2$ where $y_n \in Y_k'$}
\State  $S=S' \cup  (X_1' \setminus \{x_1\}) \cup (Y_k' \setminus \{y_n\}) $
\State \Return $S$. 
\Else
\State  $S=S' \cup (X_1' \setminus \{x_1\})$
\State \Return $S$.
\EndIf
\Else
\If{$Y_k' \geq 2$ where $y_n \in Y_k'$}
\State  $S=S' \cup (Y_k' \setminus \{y_n\})$
\State \Return $S$
\Else
\State  $S=S'$
\State \Return $S$
\EndIf
 \EndIf 
 \end{algorithmic}
  \end{algorithm}
  Therefore, $MSD$ is polynomial-time solvable on chain graphs.
\section{Inapproximability of Secure Domination Problem in Bisplit Graphs}\label{ssec:inapp_SD_bs}
	We investigate the approximation hardness of Secure domination problem in bisplit graphs. To
achieve this, we need the following result on $MSD$ in split graphs.
\begin{theorem}\label{th:iapp;sd;s}\cite{HWang_sd}
If there is some $\epsilon > 0$ such that a polynomial-time algorithm can approximate the secure domination problem for a split graph $G=(V,E)$ within a ratio of $(1 - \epsilon)~ln~|V|$, then $NP \subseteq DTIME(|V|^{O(log~log~|V|)})$.
\end{theorem}
\noindent By using  theorem \ref{th:iapp;sd;s}, we will prove similar result for $MSD$ in bisplit graphs.
\begin{theorem}\label{th:iapp;sd;bs}
If there is some $\epsilon > 0$ such that a polynomial-time algorithm can approximate the secure domination problem for a bisplit graph $G=(V,E)$ within a ratio of $(1 - \epsilon)~ln~|V|$, then $NP \subseteq DTIME(|V|^{O(log~log~|V|)})$.
\end{theorem}
\begin{proof} We establish an approximation preserving reduction from $SDD$ in split graphs to $SDD$ in bisplit graphs. Let $G=(V,E)$ be a split graph where the vertex set $V=K\cup I$. Here, $K$ induces a clique on vertex set $K=\{k_1, k_2, \dots, k_p \}$ and $I$ induces an independent set on vertex set $I=\{l_1, l_2, \dots, l_q \}$.We construct $G'= (V',E')$ from $G$ as follows. The vertex set $V(G')$ consists of original vertices of $G$, a copy of vertices of $G$ labelled as $\{m_1, m_2, \dots, m_p;$ $n_1, n_2, \dots, n_q \}$ and four additional vertices $k, l, m, n$.  Define $K_1=\{k, k_i | k_i \in K; 1\leq i \leq p \}$; $I_1=\{l, l_i | l_i \in I; 1\leq i \leq q \}$; $K_2=\{m, m_i | k_i \in K; 1\leq i \leq p \}$ and $I_2=\{n, n_i | l_i \in I; 1\leq i \leq q \}$. The edge set $E'$ is modified by removing $E(K)$ and adding additional edges. Define $E'=\{E(G)\setminus E(k)\} \cup \{\{m_i,n_j\} | \{k_i,l_j\} \in E; 1\leq i \leq p, 1\leq j \leq q\} \cup \{\{k_i,m_j\} | \forall i,j; 1\leq i,j \leq p\}\cup \{E^*\}$ where $E^*=\{\{l,k\}, \{k,m\}, \{m,n\}, \{k,m_i\}, \{m,k_i\} | \forall i; 1\leq i \leq p\}$. We note that $|V'|=2 \cdot |V(G)| +4$, $|E'|=2 \cdot |E(G)\setminus E(K)|+(p+1)^2 +2$.Thus $G'$ can be constructed in polynomial-time.
\begin{claim}\label{bs:inapp_sd}
$G'$ is a bisplit graph.
\end{claim}
\begin{proof}
We show that $G'=X \cup Y \cup Z$ is a bisplit graph where $X, Y$ and $ Z$ are three stable sets and $Y \cup Z$ forms a biclique. By our construction, $K_1 \cup K_2$ induces a biclique. Also, no two vertices of $I_1 \cup I_2$ are adjacent; hence it induces and independent set. Therefore, the graph $G'=X \cup Y \cup Z$ where $X=I_1 \cup I_2$, $Y=K_1$ and $Z=K_2$ forms a bisplit graph. 
 \end{proof}
\begin{claim}\label{npc:inapp_sd}
 $G$ has a secure dominating set $S$ with $ |S| \leq k $ if and only if $G'$ has a secure dominating set $S'$ with $|S'| \leq 2k+2$.
\end{claim}
\begin{proof}
$\textbf {\textit{Yes instance of SDD in G maps to Yes instance of SDD in G':}} $\\
Define $G^*=G' \setminus \{k, l, m, n\}$. Consider the set $S^*=S_1 \cup S_2$ where $S_1=\{k_i, l_j | k_i \in S, l_j \in S; 1\leq i \leq p, 1 \leq j \leq q\}$ and $S_2 = \{ m_i, n_j | k_i \in S, l_j \in S;  1\leq i \leq p, 1 \leq j \leq q\}$. Let $S'=S ^* \cup \{ k, m \}$ and $|S'| \leq 2k+2$. Clearly $S'$ is a dominating set of $G'$ as $S^*$ dominates all vertices of $G^*$ and $\{ k, m \} \in S'$ dominates $\{k, l, m, n\}$.
 Now we need to show $S'$ is a secure dominating set of $G'$. That is, we will prove for each vertex $v\in V(G')\setminus S'$, there exists a neighbour $v'\in S'$ such that the swap set $(S' \setminus\{ v' \})\cup \{ v \} $ is again a dominating set of $G'$.
 \begin{itemize}
 \item When $v=l$ there exists a neighbour $k\in S'$ such that $(S'\setminus\{ k \})\cup \{ l \} $ is a dominating set of $G'$. \\This is true as  $S^*$, a subset of the above set dominates all vertices of $G^*$ and $\{ l, m \}$ in the above set dominates the four other vertices.
 \item When $v=n$ there exists a neighbour $m \in S'$ such that $(S' \setminus \{ m \})\cup \{ n \} $ is a dominating set of $G'$. \\This is true as  $S^*$, a subset of the above set dominates all vertices of $G^*$ and $\{ k, n \}$ in the above set dominates the four other vertices.
 \item When $v\in V(G^*)\setminus S^*$.\\
 Since $S$ is a secure dominating set of $G$, $S_1$ securely dominates vertices of $I_1$ as adjacencies across $K$ and $I$ in $G$ are preserved in $G'$. As a consequence, $S^*$ securely dominates vertices of $I_1 \cup I_2$ and due to removal of edges $E(K)$ in $G$, we see that $S^*$ dominates vertices of $K_1 \cup K_2$. It suffices to show that, $v\in (V(G^*)\setminus S^*) \cap (K_1 \cup K_2)$ is securely dominated. Since $S^*$ is a subset of $S'$ and $S^*$ is a dominating set of $K_1 \cup K_2$, for every vertex $v\in (V(G^*)\setminus S^*) \cap (K_1 \cup K_2)$, there exists a neighbour $v' \in S^* \subseteq S'$. Therefore, for each $v \in (V(G^*)\setminus S^*) \cap (K_1 \cup K_2)$, there exists a neighbour $v'\in S'$ such that $(S' \setminus\{ v' \})\cup \{ v \} $ is a dominating set of $G'$. \\This is true as  $\{ k, m \}$ is a subset of the above set and it dominates all vertices of $K_1 \cup K_2$.
 \end{itemize}
 Therefore, $S'=S^* \cup \{ k, m \}$ is a secure dominating set of $G'$ with $|S'| \leq 2k+2$.\\
 
 \noindent $\textbf {\textit{Yes instance of SDD in G' maps to Yes instance of  SDD in G:}} $
\\ Let $G'$ have a secure dominating set $S'$ with $|S'|\leq 2k+2$. We shall prove that $G$ has a secure dominating set $S$ with $|S|\leq k$.\\
We note that, since $l$ and $n$ are pendant vertices of $G'$, $|S' \cap \{k, l\}| \geq 1$ and $|S' \cap \{m, n\}| \geq 1$. Let $G^*=G' \setminus \{k, l, m, n\}$. Consider $S^*= S' \cap V(G^*)$. $|S^*| = |S'|-[|S' \cap \{k, l\}| + |S' \cap \{m, n\}|] \leq (2k+2)-2 = 2k$. If $S^*$ securely dominates $G^*$, then we are done. Because, $S = S^* \cap (K_1 \cup I_1)$ forms a secure dominating set of $G$ and since $G^*$ is symmetric $|S| = |S^*| / 2 \leq k$. So, we assume $S^*$ is not a secure dominating set of $G^*$. Let $W \subseteq V(G^*)$ be the set of vertices that are not securely dominated by $S^*$. But, $W$ was securely dominated by $S'$ in $G'$. This implies that $W$ was securely dominated be vertices of $S' \cap \{k, l, m, n\}$ in$G^*$. Thus, by construction, $W \subseteq K_1 \cup K_2$.\\
Let  $W_1=W \cap K_1$ and $W_2=W \cap K_2$. Without loss of generality, assume $W_1\neq \emptyset$. We know that $W_1$ is securely dominated by $S'$ in $G'$, which imply $m \in S'$. By Proposition \ref{p:epn_sd}, $epn(m, S')$ is complete. It follows that $W_1$ induces a complete subgraph of $G'$ and $n \in S'$.Consequently, $|S' \cap \{m,n\}|=2$. Further, since $K_1$ is an independent set, we have $|W_1|=1$. Now, $W_2 \neq \emptyset$ or $W_2= \emptyset$. If $W_2 \neq \emptyset$, by similar argument it can be shown that $|W_2|=1$ and $|S' \cap \{k,l\}|=2$. If $W_2 = \emptyset$, then $|W_2|=0$ and $|S' \cap \{k,l\}|\geq 1$.\\
From the above, $|W|=|W_1|+|W_2| \leq |S' \cap \{k, l\}| + |S' \cap \{m, n\}|-2$. It follows that $|S^*| = |S'|-[|S' \cap \{k, l\}| + |S' \cap \{m, n\}|] \leq |S'|-[|W|+2]$. Also, $S^* \cup W$ is a secure dominating set of $G^*$ and $|S^* \cup W| = |S^*|+|W| \leq |S'|-[|W|+2]+|W|=|S'|-2 \leq 2k$. Further, since $G^*$ is symmetric $S = (S^*\cup W) \cap (K_1 \cup I_1)$ forms a secure dominating set of $G$ and $|S| = |S^* \cup W| / 2 \leq k$.
\end{proof}
\noindent Let us assume that there exists some (fixed) $\epsilon  > 0$ such that $MSD$ for bisplit graphs with $|V'|$ vertices can be approximated within a ratio of $ \alpha = (1-\epsilon)~ln~|V'|$ by a polynomial-time algorithm $\mathcal{A}$. Let $x > 0$ be a fixed integer with $x > \frac{1}{\epsilon}$. By using algorithm $\mathcal{A}$, we construct a polynomial-time algorithm for $MSD$ in split graphs as follows.
\begin{algorithm}[H]
\caption{Approx-MSD split}
\label{alg:approx-sdom;s}
\begin{algorithmic}
\Require A Split graph $G=(V,E)$
\Ensure A minimum secure dominating set $S$ of $G$
\If{there exists a minimum secure dominating set $S$ of $G$ with $|S|< x$}
    \State \Return $S$
\Else
    \State Construct the bisplit graph $G'$ as described above
    \State Compute a $MSD$ $S'$ in $G'$ using algorithm $\mathcal{A}$
    \If {$W= \emptyset$ or $epn_{G^*}(m,S') = \emptyset$}
    \State $S \gets S' \cap V$
    \State \Return $S$
    \EndIf
    \If {$W \neq \emptyset$ and $epn_{G^*}(m,S') \neq \emptyset$}
    \State $S \gets (S' \cap V) \cup \{v\}$ for some $v \in epn_{G^*}(m,S')$
    \State \Return $S$
    \EndIf
 \EndIf
\end{algorithmic}
\end{algorithm}
\noindent Initially, if there is a minimum secure dominating set $S$ of $G$ with $|S| < x$, then it can be computed in polynomial-time. Now, since the algorithm $\mathcal{A}$ runs in polynomial-time, the Algorithm \ref{alg:approx-sdom;s} also runs in polynomial-time. If the returned set $S$ satisfies $|S| < x$ then $S$ is a minimum secure dominating set of $G$ and we are done. In the following, we will analyse the case when Algorithm \ref{alg:approx-sdom;s} returned the set $S$ with $|S| \geq x$.\\ Let $S_0$ and $S_0'$ are minimum secure dominating set of $G$ and minimum secure dominating set of $G'$, respectively. By Claim \ref{npc:inapp_sd} we have $|S_0'| =2 \cdot |S_0| + 2 = 2 \cdot (|S_0| + 1)$, where $|S_0| \geq x$. Now Algorithm \ref{alg:approx-sdom;s} can compute a secure dominating set of $G$ of size $|S| \leq \frac{|S'|-2}{2} = \frac{|S'|}{2} - 1 \leq \frac{\alpha \cdot |S_0'|}{2} - 1 = \frac{\alpha \cdot 2 \cdot (|S_0| + 1)}{2} - 1 = \alpha \cdot (|S_0| + 1) - 1 < \alpha \cdot (1 + \frac{1}{|S_0|})|S_0|$. Since $|S_0| \geq x$, $|S|  < \alpha \cdot (1 + \frac{1}{x})|S_0|$. Hence, Algorithm \ref{alg:approx-sdom;s} approximates secure domination problem for given split graph $G$ within the ratio $\alpha \cdot (1 + \frac{1}{x})$. Also $x$ is a positive integer such that $x > \frac{1}{\epsilon}$. It follows that $\alpha \cdot (1 + \frac{1}{x}) \leq \alpha \cdot (1 + \epsilon) = (1-\epsilon) \cdot (1 + \epsilon)~ln~|V'| = (1-\epsilon')~ln~|V|$, where $\epsilon' = \epsilon^2$. Also, $ln~|V'| = ln~(2 \cdot |V| + 4) \approx ~ln~|V| $ for suﬃciently large value of $|V|$.\\ Therefore, Algorithm \ref{alg:approx-sdom;s} approximates secure domination problem in split graphs within a ratio of $(1-\epsilon)~ln~|V|$ for some $\epsilon  > 0$. This contradiction to Theorem \ref{th:iapp;sd;s} completes the proof.
\end{proof}
\section{Bounded tree-width graphs}\label{btw_sd}
In this section, we prove that the $MSD$ can be solved in linear-time for bounded tree-width graphs. First, we formally define the parameter $tree-width$ of a graph. For a graph $G = (V,E)$, its tree decomposition is a pair $(T,S)$, where $T = (U,F)$ is a tree, and $S = {S_u | u \in U}$ is a collection of subsets of $V$ such that
\begin{itemize}
\item $\bigcup_{u \in U} S_u = V$
\item for each $xy \in E$, there exists $u\in U$ such that $x,y \in S_u$, and
\item for all $x \in V$,the vertices in the set $\{u \in U | x \in S_u\}$ forms a subtree of $T$.
\end{itemize}
The width of a tree decomposition $(T,S)$ of a graph G is defined as $max\{ |S_u| | u \in U \}-1$. The tree-width of a graph $G$ is the minimum width of any tree decomposition of $G$. A graph is said to be a bounded tree-width graph, if its tree-width is bounded. Now, we prove that the secure domination problem can be formulated as $CMSOL$.
\begin{theorem}\label{CMSOL_sd}
For a graph $G = (V, E)$ and a positive integer $k$, the $SDD$ problem can be expressed in $CMSOL$.
\end{theorem}
\begin{proof}
Let $G = (V,E)$ be a graph and $k$ be a positive integer. The $CMSOL$ formula expressing that the existence of a dominating set $D$ of $G$ of cardinality at most $k$ is,\\
$Dom(D) = (D \subseteq V ) \land (|D| \leq k) \land ((\forall x \in V )(\exists y \in V )((y \in D) \land (x \in N[y])))$
\\Using the above $CMSOL$ formula for dominating set $D$ of cardinality at most $k$, we give $CMSOL$ formula for the secure dominating set of $G$ of cardinality at most $k$ as follows,\\
$SDom(D) = Dom(D) \land ((\forall x \in V\setminus D )(\exists y \in Dom(D)) \land ((y \in N (x)) \land Dom((D \setminus \{y\}) \cup\{x\})))$
Hence, the result follows.
\end{proof}
\noindent The famous Courcelle’s theorem \cite{COURCELLE_MSOL} states that any problem which can be expressed as a $CMSOL$ formula is solvable in linear-time for graphs having bounded tree-width. From Courcelle theorem and  theorem \ref{CMSOL_sd}, the following result directly follows.
\begin{theorem}
For bounded tree-width graphs, the $SDom$ problem is solvable in linear- time.
\end{theorem}
		\section{Conclusion and Future work:}
		We investigated the complexity of \textit{Secure Domination Problem} and proved that the problem is NP-complete on star convex and comb convex split graphs and bisplit graphs. Having found that $SDD$ on bisplit graphs is NP-complete, our next focus was to analyse the complexity of $SDD$ by restricting the cycle length. We observed that $SDD$ was polynomial-time solvable on chordal bisplit graphs and it is the presence of cycles of length four which makes the problem hard on bisplit graphs, thereby delineating the boundary between tractable and intractable cases. Also $MSD$ is polynomial-time solvable in chain graphs. We  prove that the secure domination problem is hard to be approximated for a bisplit graph within a factor of $(1-\epsilon)~ln~|V|$ for any $\epsilon > 0$, unless $NP \subseteq DTIME(|V|^{O(log~log~|V|)})$. Motivated by the established inapproximability result, approximation algorithms for secure domination problem on bisplit graphs will be investigated as part of our future work. Since $MSD$ is solvable in linear time on graphs of bounded tree-width, an intresting direction for future work is to establish a parameterized algorithm for the problem on $k-trees$, where $k$ is fixed.
	\bibliographystyle{ieeetr}
\bibliography{references} % my bib file references.bib
\end{document}